\documentclass[runningheads]{llncs}

\usepackage[T1]{fontenc}
\usepackage{graphicx}
\usepackage{amsmath,amssymb}
\usepackage{amsfonts}
\usepackage{algpseudocode} 
\usepackage{algorithm}
\usepackage{multirow}
\usepackage{url}

\usepackage[colorlinks=true,linkcolor=black,citecolor=black,urlcolor=blue,bookmarks=true]{hyperref}

\usepackage[nameinlink]{cleveref}

\DeclareMathOperator{\End}{End}
\DeclareMathOperator{\GL}{GL}

\DeclareMathOperator{\Char}{char}

\newcommand\Aut{{\mathrm{Aut}}}

\newcommand\ord{{\mathrm{ord}}}

\newcommand{\XX}{\mathcal{X}}

\def\F{{\mathbb F}}

\title{On the Classical Hardness of the Semidirect Discrete Logarithm Problem in Finite Groups}
\titlerunning{Classical Hardness of SDLP in Finite Groups}

\author{Mohammad Ferry Husnil Arif\inst{1}\orcidID{0009-0004-4965-9380} \and Muhammad Imran\inst{1}\inst{2}\orcidID{0000-0002-0776-2644}}
\institute{Universitas Indonesia, Indonesia \and University of Birmingham, UK}

\begin{document}

\raggedbottom

\maketitle

\begin{abstract}
The semidirect discrete logarithm problem (SDLP) in finite groups was proposed as a foundation for post-quantum cryptographic protocols, based on the belief that its non-abelian structure would resist quantum attacks. However, recent results have shown that SDLP in finite groups admits efficient quantum algorithms, undermining its quantum resistance. This raises a fundamental question: does the SDLP offer any computational advantages over the standard discrete logarithm problem (DLP) against classical adversaries? In this work, we investigate the classical hardness of SDLP across different finite group platforms. We establish that the group-case SDLP can be reformulated as a generalized discrete logarithm problem, enabling adaptation of classical algorithms to study its complexity. We present a concrete adaptation of the Baby-Step Giant-Step algorithm for SDLP, achieving time and space complexity $O(\sqrt{r})$ where $r$ is the period of the underlying cycle structure. Through theoretical analysis and experimental validation in SageMath, we demonstrate that the classical hardness of SDLP is highly platform-dependent and does not uniformly exceed that of standard DLP. In finite fields $\mathbb{F}_p^*$, both problems exhibit comparable complexity. Surprisingly, in elliptic curves $E(\mathbb{F}_p)$, the SDLP becomes trivial due to the bounded automorphism group, while in elementary abelian groups $\mathbb{F}_p^n$, the SDLP can be harder than DLP, with complexity varying based on the eigenvalue structure of the automorphism. Our findings reveal that the non-abelian structure of semidirect products does not inherently guarantee increased classical hardness, suggesting that the search for classically hard problems for cryptographic applications requires more careful consideration of the underlying algebraic structures.

\end{abstract}

\section{Introduction}
The presumed intractability of the discrete logarithm problem (DLP) has long been a cornerstone of modern cryptographic security, particularly through its application in the Diffie-Hellman key exchange protocol. Numerous cryptographic schemes, including key exchange, digital signatures, and public-key encryption, rely on the hardness of the DLP. However, the advent of quantum computing, particularly Shor's algorithm \cite{shor1994algorithms}, has fundamentally altered the security landscape. Shor's algorithm demonstrates that the DLP can be solved efficiently in polynomial time when implemented on a sufficiently large quantum computer, rendering classical cryptographic schemes based on the DLP vulnerable to quantum attacks.

In response to this existential threat, researchers have sought alternative approaches to construct cryptographic protocols resistant to quantum adversaries. One prominent direction is the development of new DLP variants designed to circumvent Shor's algorithm. Since the applicability of Shor's algorithm relies crucially on the algebraic structures of the underlying group, two general strategies have emerged.

The first strategy considers the DLP analogue in algebraic objects with less structure. One approach considers the DLP in finite commutative semigroups, where the absence of full group properties was initially thought to obstruct Shor's algorithm. This is based on the observation that Shor's algorithm to compute $a$ from the elements $g, h=g^a\in G$, for a finite commutative group $G$, requires the evaluation of the oracle $f(x,y)=g^xh^{-y}$. Therefore, in the case of semigroup, the oracle is not well-defined.
However, Childs and Ivanyos \cite{childs2014quantum} demonstrated that a modified version of Shor's algorithm can efficiently solve the DLP in finite commutative semigroups, undermining its viability as a quantum-resistant alternative. 

Moreover, another approach considers fully unstructured sets, leading to the proposal of the DLP analogue in commutative group actions, which appears to be the most promising approach for constructing quantum-resistant cryptographic protocols. This problem, known as the \textit{vectorization problem}, is an instance of constructive membership testing in the orbits of commutative permutation groups on large finite sets. The framework was originally introduced by Couveignes \cite{couveignes2006hard} and has since become a central problem in isogeny-based cryptography, as exemplified by CSIDH \cite{castryck2018csidh}.  

The second strategy relies on increasing the complexity of the algebraic structures underlying the DLP. First, we have the DLP analogue in finite semidirect product groups which first appears in \cite{habeeb2013public}. The problem is also known as the semidirect discrete logarithm problem (SDLP). The SDLP in the semidirect product groups $G\rtimes \End(G)$ for a finite group $G$ can be defined as follows.
\begin{problem}
    Given $g\in G, \sigma\in \End(G),$ and $h=\prod_{i=0}^{t-1}\sigma^i(g)$ for some integer $t$, determine $t$.
\end{problem}

The SDLP can be defined at different levels of generality. In its most general form, the problem considers a finite semigroup $G$ with $\sigma \in \End(G)$. When $G$ is restricted to be a finite group, we obtain the \emph{group-base case}, and when we further require $\sigma$ to be an automorphism, we have the \emph{group case} (or full group case).

The SDLP was believed to resist against Shor's algorithm and thus several post-quantum cryptographic protocols have been proposed based on the formulation of Diffie-Hellman key exchange follows from the SDLP in finite groups. Suppose two parties, Alice and Bob, agree on a public group $G$, an element $g\in G$, and an endomorphism $\sigma \in \End(G)$. Then they can arrive at the same $G-$element as follows.
\begin{itemize}
    \item[1.] Alice picks a random positive integer $x$ and computes $(g, \sigma)^x=\left(A, \sigma^x\right)$. Then, Alice sends $A=\prod_{i=0}^{x-1}\sigma^i(g)$ to Bob.
   \item[2.] Bob also picks a random  positive integer $y$, computes $(g,\sigma)^y=(B,\sigma^y)$ and sends $B=\prod_{i=0}^{y-1}\sigma^i(g)$ to Alice.
   \item[3.] Alice computes its shared key $K_A=A\sigma^x(B)$.
   \item[4.] Bob computes its shared key $K_B=B\sigma^y(A)$.
\end{itemize}
Note that $K_A=K_B$, as the following calculation shows.
\begin{eqnarray*}
\mbox{~~~~~~~~~}A\sigma^x(B)&=&
\prod_{i=0}^{x-1}\sigma^i(g)\prod_{i=0}^{y-1}\sigma^{x+i}(g)
=\prod_{i=0}^{x+y-1}\sigma^i(g)\\
&=&
\prod_{i=0}^{y-1}\sigma^i(g)\prod_{i=0}^{x-1}\sigma^{y+i}(g)\\
&=& B\sigma^y(A).
\end{eqnarray*}

However, it was recently shown in \cite{imran2024efficient,battarbee2024semidirect} that the SDLP in finite groups can be reduced to some instances that can be solved efficiently by Shor's algorithm and well-known generalizations.

More general variant of the DLP is proposed within the framework of general finite non-abelian groups. This was first introduced by Stickel in \cite{stickel2005new}. The DLP analogue in general finite non-abelian groups $G$ is defined as follows.
\begin{problem}
    Given $u,v,w,x\in G$ such that $w=u^axv^b$ for some integer $a$ and $b$, determine $a$ and $b$.
\end{problem}
The problem serves as the basis for Stickel's Diffie-Hellman-like key exchange protocol and has also been adapted in tropical cryptography, notably first proposed by Grigoriev and Shpilrain in \cite{grigoriev2014tropical}.

Since the SDLP in finite groups is efficiently solvable by quantum algorithms, this raises a question whether this problem is even considerably harder than the standard DLP or asymptotically living in the same classical complexity classes. The reduction from group-base to group case is classical and efficient, suggesting that the additional structure of the semidirect product might not provide significant classical hardness benefits. However, unlike the DLP analogue in general non-abelian groups, the SDLP appears quite different from the standard DLP. Therefore, it was not clear how to adopt the standard classical algorithms for DLP to solve the SDLP.

The first result of this work establishes that the SDLP in finite groups can be equivalently expressed in the form of the DLP analogue in general non-abelian groups in Problem 2, which is a form that more naturally generalizes the standard DLP. This reformulation allows us to adapt the \textit{Baby-step Giant-step} algorithm to the SDLP in finite groups, which we implement in SageMath. Finally, we provide a detailed complexity comparison of these algorithms for the standard DLP and the SDLP across several classes of finite groups, including the multiplicative group of finite fields $\F_p^*$, the group of rational points of elliptic curves $E(\F_p)$ over finite fields $\F_p$, and elementary abelian groups $\F_p^n$.

In this work, we focus specifically on the group setting for several reasons. First, classical algorithms like Baby-Step Giant-Step are most naturally formulated for groups where inverses exist. Second, and most importantly, \cite{imran2024efficient} established an efficient classical reduction from the group-base case to the group case, allowing us to restrict our attention to the case where $\sigma$ is an automorphism without loss of generality for group-based instances.

\section{Classical algorithms for the SDLP}\label{sec:classical-algorithms}

In this section, we systematically develop classical algorithmic approaches for solving the semidirect discrete logarithm problem (SDLP). Our exposition follows a deliberate progression through several key components. First, we establish a reduction from the general group-base case SDLP to the group case, demonstrating that we can focus our algorithmic efforts on the latter without loss of generality. Next, we analyze the structural properties of the group case SDLP, particularly examining the cycle structure and period bounds that constrain potential solutions. We then demonstrate a critical reduction from the SDLP to a generalized discrete logarithm problem (GDLP) formulation, which enables us to adapt classical DLP algorithms to the SDLP context. Building on this foundation, we present a modified Baby-step Giant-step algorithm specifically adapted for the SDLP and analyze its complexity in terms of the period of the underlying cycle structure. This comprehensive treatment establishes a framework for comparing the computational difficulty of the SDLP relative to the standard DLP across various group structures, which we explore in subsequent sections.

\subsection{Reduction from group-base case to group case}

The reduction follows directly from \cite{imran2024efficient}. Let $G$ be a finite group with $\sigma$ an endomorphism of $G$. To reduce the problem to the automorphism case, consider $K=\cap_{i=0}^\infty \sigma^i(G)$ and let $k_0$ be the smallest non-negative integer such that $K=\sigma^{k_0}(G)$. It is obvious that $k_0 \leq  \lceil \log|G| \rceil $. So we can just choose $k$ as the upper bound of $\lceil \log|G| \rceil$ and get $K = \sigma^{k}(G)$. Then the restriction of $\sigma$ to $K$ is an automorphism. 

Let $t$ be solution for SDLP$(G,\sigma)$ given by $g$ and $h$, so we get $h = \prod_{i=0}^{t-1}\sigma^i(g)$. Moreover

\[\sigma^k(h) = \sigma^k \left(\prod_{i=0}^{t-1}\sigma^i(g) \right) = \prod_{i=0}^{t-1}\sigma^i \left( \sigma^k (g) \right) .\]

It follows that $t$ also is the solution for SDLP$(K,\sigma |_K)$ given by $\sigma^k(g)$ and $\sigma^k(h)$. This gives a classical polynomial-time reduction from the group-base case SDLP to the group case SDLP. Throughout the rest of this work we will focus only on the group case SDLP.

\subsection{Group case SDLP}

The structure of the semidirect product group provides important insights into the properties of the SDLP. Following the approach in \cite{battarbee2023spdh}, we first introduce some key concepts to understand the cycle structure.

\begin{definition}\label{def:s_function}
Let $G$ be a finite group. Fix some $(g, \sigma) \in G \rtimes \Aut(G)$. For any $x \in \mathbb{Z}$, the function $s_{g,\sigma} : \mathbb{Z} \rightarrow G$ is defined as the group element such that

$$(g, \sigma)^x = (s_{g,\sigma}(x), \sigma^x).$$
\end{definition}

Recall the binary operation in the semidirect product $G\rtimes \Aut(G)$ is defined as $(g_1, \phi_1)(g_2,\phi_2)=(g_1\phi_1(g_2), \phi_1\phi_2)$ for any $(g_1, \phi_1),(g_2,\phi_2)\in G \rtimes \Aut(G)$.
\begin{definition}\label{def:cycle_period}
Let $G$ be a finite group and $(g, \sigma) \in G \rtimes \Aut(G)$. The set $$ \XX_{g,\sigma} := \{s_{g,\sigma}(i) : i \in \mathbb{Z}\}.$$ 
is called the cycle of $(g, \sigma)$, and its size is called the period of $(g, \sigma)$.
\end{definition}

The work of \cite{battarbee2023spdh} gives some important result. Let $r$ be the period of $\XX_{g,\sigma}$. Then we have $\XX_{g,\sigma} = \{1, g, \cdots, s_{g,\sigma}(r-1)\}$. Consequently, the solution of SDLP$(G,\sigma)$ will be unique up to modulo $r$. Also, every $k$ such that $s_{g,\sigma}(k)=1$ we will have $r \mid k$. Let $n$ be of the order $\sigma$. Because $(g, \sigma) \in G \rtimes \langle \sigma \rangle$, we have $s_{g,\sigma}(n|G|) = 1$. This means $r$ is one of the factors of $n|G|$. But we can actually find much lower bound than this value for most case platform groups that we will study.

Define fixed-point subgroup $G^\sigma := \{g \in G \mid \sigma(g)=g \}$. 
We show that $s_{g,\sigma}(n) \in G^\sigma$.

\begin{theorem}\label{thm:sn-in-fixed}
Let $G$ be a finite group and $\sigma$ is automorphism with order $n$. Then $s_{g,\sigma}(n) \in G^\sigma$.
\end{theorem}

\begin{proof}

Suppose $g' := s_{g,\sigma}(n)$. We have

$$\sigma(g') = \sigma(s_{g,\sigma}(n)) = \prod_{i=1}^n\sigma^i(g) = \prod_{i=0}^{n-1}\sigma^i(g) = s_{g,\sigma}(n) = g'.$$

\end{proof}

Now we are ready to show the main result.

\begin{theorem}\label{thm:sn-equal-1}
Let $G$ be a finite group and $\sigma$ is automorphism with order $n$. Then $s_{g,\sigma}(\ord(s_{g,\sigma}(n)) \cdot n) = 1$.
\end{theorem}

\begin{proof}
Since $\sigma$ has order $n$, we have $\sigma^{in} = \text{id}$ for all $i \in \mathbb{Z}$. By \cite[Theorem 1]{battarbee2023spdh}, we have

$$s_{g,\sigma}(kn) = \prod_{i=0}^{k-1} \sigma^{in}(s_{g,\sigma}(n)) = \prod_{i=0}^{k-1} s_{g,\sigma}(n) = s_{g,\sigma}(n)^k.$$

Setting $k = \ord(s_{g,\sigma}(n))$ gives $s_{g,\sigma}(\ord(s_{g,\sigma}(n)) \cdot n) = s_{g,\sigma}(n)^{\ord(s_{g,\sigma}(n))} = 1$.
\end{proof}

In many of the group cases considered, $\ord(s_{g,\sigma}(n))$ is substantially smaller than $|G|$, resulting in a considerably tighter bound on $r$.

\subsection{Reduction from SDLP to GDLP}\label{equiv}
To efficiently solve the SDLP using classical methods, it is helpful to reformulate the problem in terms of a more general group-theoretic framework. In this subsection, we show that every instance of the group-case SDLP can be equivalently expressed as a GDLP in the semidirect product group $G \rtimes \text{Aut}(G)$, where the group operation incorporates the action of automorphisms.

\begin{theorem}\label{thm:red-sdlp-gdlp}
    Let $G$ be a finite group and $\sigma \in \Aut(G)$. Given $g, h \in G$ such that
\[h = \prod_{i=0}^{t-1} \sigma^i(g) = s_{g,\sigma}(t),\]
then the SDLP instance defined by $(G, \sigma, g, h)$ is equivalent to solving the following GDLP: find $t \in \mathbb{Z}$ such that $w = u^t v^t$ for
\[
u = (g, \sigma), \quad v = (1, \sigma^{-1}), \quad w = (h, \text{id}) \in G \rtimes \text{Aut}(G).
\]
\end{theorem}
\begin{proof}
Recall the group operation in the semidirect product $G \rtimes \Aut(G)$ is defined by:
\[(g_1, \phi_1) \cdot (g_2, \phi_2) = (g_1 \cdot \phi_1(g_2), \phi_1 \circ \phi_2).\]
Therefore, we have $u^t = (g, \sigma)^t = (s_{g,\sigma}(t), \sigma^t)$ and  $v^t = (1, \sigma^{-1})^t = (1, \sigma^{-t})$.
Thus,
\[u^t v^t = (s_{g,\sigma}(t), \sigma^t) \cdot (1, \sigma^{-t}) = (s_{g,\sigma}(t), \text{id}) = (h, \text{id}) = w.\]
Therefore, solving SDLP is equivalent to finding $t$ such that $u^t v^t = w$ in the group $G \rtimes \Aut(G)$.
\end{proof}
This reduction is significant for two main reasons:
\begin{enumerate}
    \item Algorithmic leverage: It allows us to view the SDLP through the lens of group exponentiation in a non-abelian setting, opening the door to adapting known discrete logarithm algorithms (e.g., Baby-Step Giant-Step, Pollard's rho) to the semidirect product framework.
    \item Structural clarity: The decomposition of SDLP into GDLP isolates the nonlinear part of the problem (the recursive application of automorphisms) into a group-theoretic form that is more amenable to collision-based methods.
\end{enumerate}

Moreover, this reformulation highlights the similarity between SDLP and the Stickel-type discrete log problems in non-abelian groups, as introduced in \cite{stickel2005new}. In both cases, the logarithmic exponent is recovered not through simple group exponentiation, but through the interaction of two operations—here captured by $u^t$ and $v^t$ respectively—that encode algebraic structure beyond commutativity.

This equivalence sets the stage for \Cref{sec:bsgs-sdlp}, where we present a concrete adaptation of the Baby-Step Giant-Step algorithm to solve this GDLP efficiently. The insight from \Cref{thm:red-sdlp-gdlp} is key to this adaptation, as it ensures that every SDLP instance can be encoded as a GDLP in a unified and structurally sound manner.

\subsection{Baby-step giant-step adaptation}\label{sec:bsgs-sdlp}
The Baby-step Giant-step (BSGS) algorithm, originally introduced by Shanks \cite{Shanks1971ClassNA}, is a classical collision-based method for solving the standard DLP in cyclic groups. In this section, we present a generalized adaptation of the BSGS algorithm designed to solve the SDLP, using the reduction to the GDLP established in \Cref{thm:red-sdlp-gdlp}.

Let $u = (g, \sigma)$, $v = (1, \sigma^{-1})$, and $w = (h, \text{id}) \in G \rtimes \text{Aut}(G)$, where $g, h \in G$, and $\sigma \in \text{Aut}(G)$. The SDLP then reduces to finding $t \in \mathbb{Z}$ such that $w = u^tv^t.$

Let $r$ be the period of the cycle $X_{g,\sigma}$, which bounds the search space for $t$. Set $m = \lceil \sqrt{r} \rceil$. Then, any $t < r$ can be uniquely expressed in the form $t = im + j$ for integers $0 \leq i, j < m$.

We now derive the core identity exploited by the adapted BSGS algorithm. From the group structure of $G \rtimes \text{Aut}(G)$, we have:

\[
w = u^t v^t = u^{im + j} v^{im + j} = u^{im} u^{j} v^{j} v^{im},
\]

which can be rearranged as:

\[
u^{j} v^{j} = u^{-im} w v^{-im}.
\]

This identity implies that a collision between the two sides over their respective ranges yields a solution for $t$. Thus, we construct two sets
\begin{enumerate}
    \item Baby steps: $\{ u^j v^j : 0 \leq j < m \}$.
    \item Giant steps: $\{ u^{-im} w v^{-im} : 0 \leq i < m \}$
\end{enumerate}
A match between an element from each set yields
\[u^j v^j = u^{-im} w v^{-im} \implies t = im + j.\]

\begin{algorithm}[H]
\caption{BSGS for GDLP}\label{alg:bsgs-sdlp}
\begin{algorithmic}[1]
\Require $u,v,w \in G \rtimes \Aut(G)$ and $r$ is period of $\XX_{g,\sigma}$
\Ensure $t$ such that $u^tv^t=w$
\State Initialize empty hash table $\mathcal{T}$
\State $m \gets \lceil \sqrt{r} \rceil$
\State $temp \gets (1, \text{id})$
\For{$j = 0$ \textbf{to} $m-1$}
    \State $\mathcal{T}[temp] \gets j$ \Comment{Baby Step}
    \State $temp \gets u \cdot temp \cdot v$
\EndFor
\State Precompute $u^{-m}, v^{-m}$
\State $temp \gets w$
\For{$i = 0$ \textbf{to} $m-1$}
    \If{$temp \in \mathcal{T}$} \Comment{Giant Step}
        \State $j \gets \mathcal{T}[temp]$
        \State \Return $i \cdot m + j$
    \Else
        \State $temp \gets u^{-m} \cdot temp \cdot v^{-m}$
    \EndIf
\EndFor
\end{algorithmic}
\end{algorithm}

\begin{theorem}
    Let $u,v,w \in G\rtimes \Aut(G)$ be such that $u^tv^t=w$ and $r$ is period of $\XX_{g,\sigma}$. Assuming $u^{-1}$ and $v^{-1}$ are efficiently computable, we can find $t$ with runtime $O(\sqrt{r})$ and space $O(\sqrt{r})$ using Algorithm~\ref{alg:bsgs-sdlp}.
\end{theorem}

\begin{proof}
    The algorithm correctness follows from the equality $u^{j} v^{j} = u^{-im} w v^{-im}$ when $t = im + j$. Since we know that $0 \leq t < r$, and we have partitioned this range with $0 \leq j < m$ and $0 \leq i < m$ where $m = \lceil \sqrt{r} \rceil$, we will find a solution if one exists. The algorithm achieves $O(\sqrt{r})$ time complexity and $O(\sqrt{r})$ space complexity because it performs exactly $m$ group operations in the baby step phase while storing these values in a hash table, followed by at most $m$ group operations in the giant step phase with constant-time hash table lookups.
\end{proof}

When applying this algorithm to SDLP instances via the reduction in Theorem~\ref{thm:red-sdlp-gdlp}, we have $u = (g, \sigma)$, $v = (1, \sigma^{-1})$, and $w = (h, \text{id})$. In this case, computing $u^{-1} = (\sigma^{-1}(g^{-1}), \sigma^{-1})$ and $v^{-1} = (1, \sigma)$ requires efficient computation of $\sigma^{-1}$. For the platform groups studied in Section~\ref{sec:finite-field}, \ref{sec:elliptic-curve}, and \ref{sec:abelian-group}, this requirement is satisfied as we discuss in the respective subsections.

This algorithm forms the computational foundation for our experimental analysis in later sections. Its performance, as we will see, depends heavily on the structure of the underlying platform group and the properties of the automorphism $\sigma$. The period $r$ plays a crucial role in determining the efficiency of solving SDLP via this method.

\section{Specific platform groups}

In this section, we analyze the computational complexity of both the discrete logarithm problem (DLP) and the SDLP across several fundamental platform groups that are commonly employed in cryptographic applications. Our investigation focuses on three platform groups: the multiplicative group of finite fields $\mathbb{F}_p^*$, elliptic curves over finite fields $E(\mathbb{F}_p)$, and elementary abelian groups $\mathbb{F}_p^n$.

As established in \Cref{sec:classical-algorithms}, we focus specifically on the group case of the SDLP where $\sigma$ is an automorphism, leveraging the polynomial-time reduction from the group-base case. For each platform, we derive complexity bounds for both problems, examining how the inherent algebraic structure influences the computational difficulty. Our analysis reveals that the relative hardness between DLP and SDLP varies dramatically depending on the underlying group structure, with some platforms favoring one problem over the other.

The results of our analysis are summarized in \Cref{tab:complexity-comparison}, which provides a comprehensive comparison of asymptotic complexities for the group case SDLP.

\begin{table}[h]
\centering
\caption{Complexity comparison of DLP and SDLP across different platform groups}
\label{tab:complexity-comparison}
\begin{tabular}{|c|c|c|}
\hline
\textbf{Platform Group} & \textbf{DLP} & \textbf{SDLP (group case)} \\
\hline
$\mathbb{F}_p^*$ & $e^{((64/9)^\frac{1}{3} + o(1))(\ln p)^{\frac{1}{3}}(\ln \ln p)^{\frac{2}{3}}}$ & $\mathcal{O}(p)$\footnotemark \\
\hline
$E(\mathbb{F}_p)$ & $\mathcal{O}(\sqrt{p})$ & $\mathcal{O}(1)$ \\
\hline
\multirow{2}{*}{$\mathbb{F}_p^n$} & \multirow{2}{*}{$\mathcal{O}(1)$} & $\mathcal{O}(p^{n/2})$ (no eigenvalue 1) \\
\cline{3-3}
 & & $\mathcal{O}(p^{\frac{n}{2} + 1})$ (with eigenvalue 1) \\
\hline
\end{tabular}
\end{table}
\footnotetext{Precise complexity: $\mathcal{O}(\sqrt{(p-1) \cdot \gcd(k-1, p-1)})$ for automorphism $x \mapsto x^k$. For random $k \in \mathbb{Z}_{p-1}^*$, $\gcd(k-1, p-1)$ is typically small on average.}

These findings reveal several interesting patterns. The relationship between DLP and SDLP complexity is highly platform-dependent: while both problems exhibit comparable difficulty in finite fields $\mathbb{F}_p^*$, the SDLP becomes trivial in elliptic curves but can be harder in elementary abelian groups.

\subsection{Multiplicative group of finite field}\label{sec:finite-field}

Let $\mathbb{F}_p$ be a finite field of prime order $p$. The multiplicative group $\mathbb{F}_p^*$ is cyclic of order $p-1$, which completely determines its automorphism structure. Since any automorphism of a cyclic group is uniquely determined by where it sends a generator, and must map the generator to another generator to maintain bijectivity, we have $\Aut(\mathbb{F}_p^*) \cong \mathbb{Z}_{p-1}^*$. 

Specifically, each automorphism of $\mathbb{F}_p^*$ has the form $\phi_k: x \mapsto x^k$ where $k \in \mathbb{Z}_{p-1}^*$ (i.e., $\gcd(k, p-1) = 1$). The requirement that $\gcd(k, p-1) = 1$ ensures that $\phi_k$ is bijective: if $g$ is a generator of $\mathbb{F}_p^*$, then $g^k$ is also a generator if and only if $k$ is coprime to $p-1$. This characterization of automorphisms will be crucial for analyzing the SDLP complexity in this group.

For computational purposes, we note that computing $\sigma^{-1}$ is efficient in this setting. Given an automorphism $\phi_k: x \mapsto x^k$ where $k \in \mathbb{Z}_{p-1}^*$, the inverse automorphism is $\phi_{k^{-1}}: x \mapsto x^{k^{-1} \bmod (p-1)}$. The modular inverse $k^{-1} \bmod (p-1)$ can be computed in $O(\log p)$ time using the Extended Euclidean Algorithm, making $\sigma^{-1}$ efficiently computable as required by Algorithm~\ref{alg:bsgs-sdlp}.

Now we are ready to compare the standard DLP and the SDLP for the multiplicative group of finite fields $\mathbb{F}_p^*$. 
\subsubsection{DLP}
For the multiplicative group of finite fields $\mathbb{F}_p^*$, there are several classical algorithms available for solving the DLP. The generic approaches include BSGS \cite{Shanks1971ClassNA} and Pollard's rho algorithm \cite{pollard78dl}, both with time complexity $\mathcal{O}(\sqrt{p-1}) = \mathcal{O}(\sqrt{p})$.

However, the best classical algorithm for solving the DLP is the Number Field Sieve (NFS), which was originally developed for integer factorization \cite{lenstra90nfs,buhler93nfs} and later adapted to the discrete logarithm problem \cite{gordon93nfs}.

For a prime $p$, the NFS algorithm achieves a subexponential time complexity of $e^{((64/9)^\frac{1}{3} + o(1))(\ln p)^{\frac{1}{3}}(\ln \ln p)^{\frac{2}{3}}}$. The algorithm works by constructing number fields and using relations between elements to build a system of linear equations in a large but sparse matrix. The solution to this system gives the discrete logarithm.

\subsubsection{SDLP}
The semidirect product $\mathbb{F}_p^* \rtimes \Aut(\mathbb{F}_p^*)$ has elements of the form $(a, \phi_k)$ where $a \in \mathbb{F}_p^*$ and $\phi_k \in \Aut(\mathbb{F}_p^*)$. The group operation is defined as:
\[(a, \phi_k) \cdot (b, \phi_l) = (a \cdot \phi_k(b), \phi_k \circ \phi_l) = (a \cdot b^k, \phi_{kl \bmod p-1}).\]
The identity element is $(1, \phi_1)$, and the inverse of $(a, \phi_k)$ is $(\phi_{k^{-1}}(a^{-1}), \phi_{k^{-1}})$ where $k^{-1}$ is computed modulo $p-1$.

To analyze the SDLP complexity, we apply \Cref{thm:sn-in-fixed} and \Cref{thm:sn-equal-1}. Let $n = \text{ord}(k)$ as an element in $\mathbb{Z}_{p-1}^*$. By \Cref{thm:sn-in-fixed}, we have $a' := s_{a,\phi_k}(n) \in (\mathbb{F}_p^*)^{\phi_k}$, the fixed-point subgroup under $\phi_k$. We need to compute the size of this subgroup:
\[(\mathbb{F}_p^*)^{\phi_k} = \{x \in \mathbb{F}_p^* \mid \phi_k(x) = x\} = \{x \in \mathbb{F}_p^* \mid x^k = x\} = \{x \in \mathbb{F}_p^* \mid x^{k-1} = 1\}.\]

The size of this subgroup is precisely $|(\mathbb{F}_p^*)^{\phi_k}| = \gcd(k-1, p-1)$, since these are exactly the elements whose order divides $k-1$ in the cyclic group $\mathbb{F}_p^*$. Therefore, $\text{ord}(a') \leq \gcd(k-1, p-1)$. By \Cref{thm:sn-equal-1}, the period $r$ of the cycle $\mathcal{X}_{a,\phi_k}$ divides $\text{ord}(a') \cdot n$, giving us the bound:
\[r \leq \text{ord}(a') \cdot n \leq \gcd(k-1, p-1) \cdot \phi(p-1).\]

In the worst case, when $\gcd(k-1, p-1)$ is maximized, we have $r = \mathcal{O}(p \cdot \phi(p-1)) = \mathcal{O}(p^2)$. Applying \Cref{alg:bsgs-sdlp}, the time complexity of SDLP becomes $\mathcal{O}(\sqrt{p^2}) = \mathcal{O}(p)$.

\begin{remark}
It is worth noting that SDLP in $\mathbb{F}_p^*$ generally admits more efficient algorithms using \Cref{alg:bsgs-sdlp} if $k$ chosen randomly because value of $\gcd(k-1, p-1)$ will be small on average the value of common divisor of random number $k-1$ and $p-1$. In practice as later will shown on \Cref{sec:simulation}, we will have time complexity $\mathcal{O} \left(\sqrt{\phi(p-1)} \right) = \mathcal{O}(\sqrt{p})$ as compared to $\mathcal{O}(p)$. 
\end{remark}

\subsection{Elliptic curve}\label{sec:elliptic-curve}

Let $E$ be an elliptic curve over the finite field $\mathbb{F}_p$ where $p > 3$. The curve can be given by the Weierstrass equation $y^2 = x^3 + ax + b$ with $a, b \in \mathbb{F}_p$ and $4a^3 + 27b^2 \neq 0$. The set $E(\mathbb{F}_p)$ of $\mathbb{F}_p$-rational points forms an abelian group under the chord-and-tangent law.

The endomorphism ring $\End(E)$ consists of all group homomorphisms from $E$ to itself that are defined over $\overline{\mathbb{F}}_p$. The automorphism group $\Aut(E)$ comprises all invertible elements in $\End(E)$. A fundamental result states that $|\Aut(E)|$ divides 24, with the exact order determined by the $j$-invariant of $E$ (see \cite{silverman2009arithmetic}):

\begin{center}
\begin{tabular}{ |c|c|c| } 
 \hline
 $|\Aut(E)|$ & $j(E)$ & $\Char(\mathbb{F}_p)$ \\
 \hline
 2 & $j(E)\neq 0, 1728$ & any \\ 
 4 & $j(E)=1728$ & $p \neq 2, 3$ \\ 
 6 & $j(E)=0$ & $p \neq 2, 3$ \\ 
 12 & $j(E)=0$ or $1728$ & $p = 3$\\ 
 24 & $j(E)=0$ or $1728$ & $p = 2$ \\ 
 \hline
\end{tabular}
\end{center}

This bounded size of $\Aut(E)$ is crucial for our analysis of the SDLP complexity on elliptic curves. Additionally, the automorphisms of elliptic curves have explicit descriptions and bounded order (at most 24), making their inverses readily computable.

\subsubsection{DLP}

Unlike the DLP in multiplicative groups of finite fields, the best known classical algorithms for solving the (Elliptic Curve Discrete Log Problem) ECDLP are generic algorithms such as Baby-Step Giant-Step \cite{Shanks1971ClassNA} and Pollard's rho algorithms \cite{pollard78dl}, both with time complexity $\mathcal{O}(\sqrt{\#E(\mathbb{F}_p)}) = \mathcal{O}(\sqrt{p})$ (By Hasse Bound). This suggests that the ECDLP is genuinely harder than the DLP in finite fields of comparable size, which is a key reason for the popularity of elliptic curve cryptography.

\subsubsection{SDLP}
The semidirect product $E(\mathbb{F}_p) \rtimes \Aut(E(\mathbb{F}_p))$ has elements of the form $(P, \alpha)$ where $P \in E(\mathbb{F}_p)$ and $\alpha \in \Aut(E)$. The group operation in this semidirect product is defined as:
\[(P, \alpha) \cdot (Q, \beta) = (P + \alpha(Q), \alpha \circ \beta).\]
The identity element is $(\mathcal{O}, [1])$. The inverse of $(P, \alpha)$ is $(-\alpha^{-1}(P), \alpha^{-1})$.

Instead using \Cref{thm:sn-equal-1}, we can show that the bound $r$ is constant using the fact that $\alpha$ is element of endomorphism ring $\End(E)$.

\begin{theorem}\label{thm:ec-sn}
Let $P \in E(\mathbb{F}_p)$ and $\alpha \in \Aut(E)$ is non-trivial automorphism with order $n$, then $s_{P, \alpha}(n) = \mathcal{O}$
\end{theorem}

\begin{proof}
We have that 
\[s_{P, \alpha}(n) = P + \alpha(P) + \cdots + \alpha^{n-1}(P) = ([1] + \alpha + \cdots + \alpha^{n-1})(P).\]
Also notice that
\begin{align*}
\alpha^n &= [1] \\
\alpha^n - [1] &= [0] \\
(\alpha - [1])(\alpha^{n-1} + \cdots + \alpha + [1]) &= [0] \\
\Longrightarrow \alpha^{n-1} + \cdots + \alpha + [1] &= [0]
\end{align*}

By assumption $\alpha$ is non-trivial automorphism. Finally, we have
$([1] + \alpha + \cdots + \alpha^{n-1})(P) = [0]P = \mathcal{O}$.
\end{proof}

As a result, by \Cref{thm:ec-sn} we have the period $r$ of the cycle $\XX_{P,\alpha}$ divide $n$ which is at most 24. So the complexity for solving SDLP is $\mathcal{O}(1)$.

\subsection{Elementary abelian group}\label{sec:abelian-group}

An elementary abelian $p$-group of rank $n$ is a group isomorphic to $\mathbb{F}_p^n$, where every non-identity element has order $p$. As $\mathbb{F}_p^n$ is an $n$-dimensional vector space over $\mathbb{F}_p$, its automorphism group is isomorphic to $\GL_n(\mathbb{F}_p)$, the general linear group of $n \times n$ invertible matrices over $\mathbb{F}_p$. Under this isomorphism, each automorphism $\sigma \in \Aut(\mathbb{F}_p^n)$ corresponds to a matrix $A_\sigma \in \GL_n(\mathbb{F}_p)$ such that $\sigma(v) = A_\sigma v$ for all $v \in \mathbb{F}_p^n$. A well-known result states that the maximum order of an element in 
$\GL_n(\mathbb{F}_p)$ is $p^n - 1$ \cite{darafsheh2005}, achieved by 
companion matrices of primitive polynomials of degree $n$ over $\mathbb{F}_p$.

Under this isomorphism, computing $\sigma^{-1}$ for an automorphism $\sigma$ corresponds to computing the inverse of the matrix $A_\sigma \in \GL_n(\mathbb{F}_p)$. Matrix inversion can be performed in $O(n^3)$ operations over $\mathbb{F}_p$ using Gaussian elimination, which is polynomial time. Therefore, $\sigma^{-1}$ is efficiently computable as required by Algorithm~\ref{alg:bsgs-sdlp}.

\subsubsection{DLP}

For elementary abelian groups, the discrete logarithm problem takes a particularly simple form. Given $v \in \mathbb{F}_p^n$ and $tv \in \mathbb{F}_p^n$ for some scalar $t$, we need to determine $t$. Unlike the situation in general groups, where no direct formula exists, in a vector space structure like $\mathbb{F}_p^n$, we can solve this problem directly.

Specifically, since $v$ and $tv$ are both vectors, we can select any non-zero component of $v$, say $v_i \neq 0$, and compute $t = (tv)_i \cdot (v_i)^{-1} \mod p$. This direct computation can be performed in constant time, making the complexity $\mathcal{O}(1)$. This is substantially more efficient than the DLP in other group structures such as $\mathbb{F}_p^*$ or $E(\mathbb{F}_p)$.

\subsubsection{SDLP}

The semidirect product $\mathbb{F}_p^n \rtimes \Aut(\mathbb{F}_p^n)$ has elements of the form $(v, A)$ where $v \in \mathbb{F}_p^n$ and $A \in \GL_n(\mathbb{F}_p)$, with the group operation:
\[(v_1, A_1) \cdot (v_2, A_2) = (v_1 + A_1v_2, A_1A_2).\]

The identity element is $(0, I)$ where $I$ is the identity matrix in $\GL_n(\mathbb{F}_p)$, and the inverse of $(v, A)$ is $(-A^{-1}v, A^{-1})$.

For the semidirect discrete logarithm problem in elementary abelian groups, we need to analyze how the complexity depends on the properties of the automorphism matrix. As established earlier, for a matrix $A \in \GL_n(\mathbb{F}_p)$ the maximum order is $p^n - 1$.

The complexity of SDLP varies based on whether the automorphism matrix $A$ has an eigenvalue of 1. When $A$ has no eigenvalue equal to 1, the fixed-point subgroup $(\mathbb{F}_p^n)^A = \{v \in \mathbb{F}_p^n \mid Av = v\}$ contains only the zero vector, since $Av = v$ if and only if $v = 0$. Let $\ell = \text{ord}(A)$. By \Cref{thm:sn-in-fixed}, we have $s_{v,A}(\ell) \in (\mathbb{F}_p^n)^A$. Since the only element in this fixed-point subgroup is the zero vector, we conclude that $s_{v,A}(\ell) = 0$. By \Cref{thm:sn-equal-1}, this implies that the period $r$ of the cycle $\mathcal{X}_{v,A}$ divides $\text{ord}(s_{v,A}(\ell)) \cdot \ell = \text{ord}(0) \cdot \ell = \ell$, which is at most $p^n - 1$. Using \Cref{alg:bsgs-sdlp}, we achieve a complexity of $\mathcal{O}(\sqrt{r}) = \mathcal{O}(\sqrt{p^n - 1}) = \mathcal{O}(p^{n/2})$.

Otherwise, when $A$ has an eigenvalue of 1, the fixed-point subgroup $(\mathbb{F}_p^n)^A$ is non-trivial. With high probability $s_{v,A}(\ell) \ne 0$. But since every non-zero element in $\mathbb{F}_p^n$ has order $p$, by \Cref{thm:sn-equal-1} we have $s_{v,A}(p\ell) = 0$. The worst-case bound for the period $r$ becomes $p \cdot \ord(A) \le p \cdot (p^n - 1)$. Using \Cref{alg:bsgs-sdlp}, the complexity increases to $\mathcal{O}(\sqrt{p \cdot (p^n - 1)}) = \mathcal{O}(p^{\frac{n}{2} + 1})$.

\section{Simulation results}\label{sec:simulation}
To validate our theoretical analysis and provide practical insights into the relative hardness of the DLP and SDLP across different platform groups, we implemented the algorithms in SageMath and conducted systematic experiments. Our implementation is publicly available at \url{https://github.com/swusjask/classical-sdlp}. We used the Baby-Step Giant-Step algorithm for all test cases to ensure a fair comparison of the inherent structural properties of each group, even though more efficient algorithms exist for specific platforms (such as the Number Field Sieve for finite fields).

Our experiments measured both the theoretical bounds (group order for DLP, cycle period for SDLP) and the actual computation time required by the BSGS algorithm. For finite fields and elliptic curves, we tested primes ranging from $2^{30}$ to $2^{37}$ to cover cryptographically relevant sizes. For elementary abelian groups, we fixed the dimension $n=3$ and varied the prime from $2^6$ to $2^{12}$, specifically examining the impact of eigenvalue structure on SDLP complexity.

\subsection{Finite Fields and Elliptic Curves}
We also compared the DLP and SDLP across both finite fields $\mathbb{F}_p^*$ and elliptic curves $E(\mathbb{F}_p)$ for primes $p$ ranging from $2^{30}$ to $2^{37}$, focusing on cryptographically relevant sizes. \Cref{fig:field-curve} shows both the comparison of group orders/cycle periods and their corresponding computation times.

\begin{figure}[h]
\centering
\includegraphics[width=\textwidth]{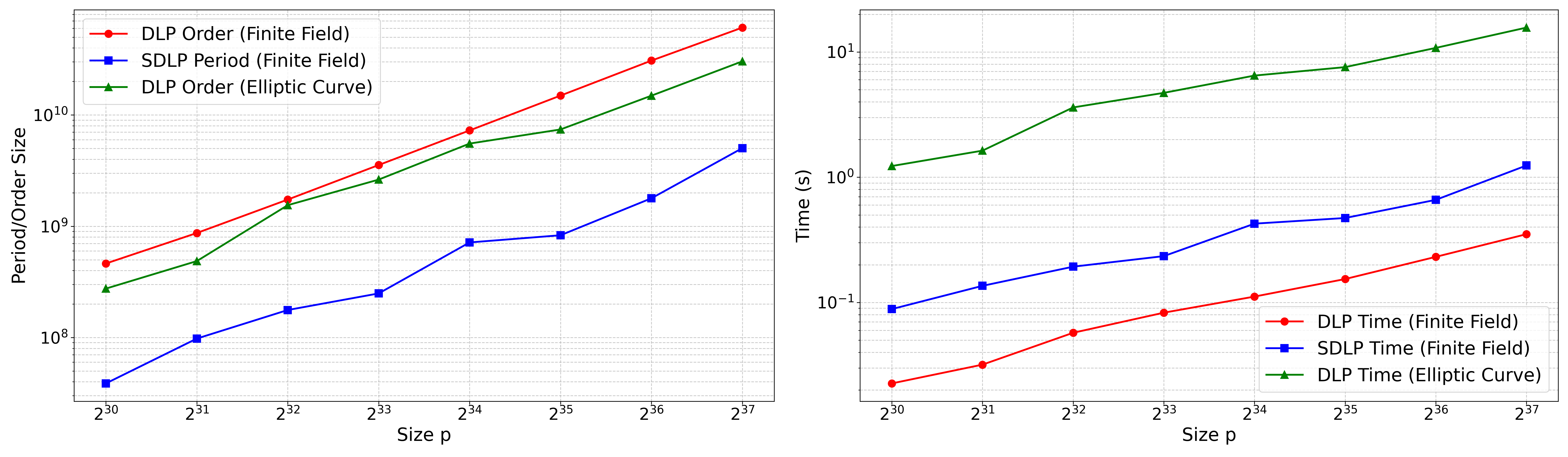}
\caption{Comparison of order/period sizes and computation times for DLP and SDLP in finite fields and elliptic curves.}
\label{fig:field-curve}
\end{figure}

Several interesting observations emerge from this experiment. The orders for both DLP in finite fields ($p-1$) and DLP in elliptic curves (approximately $p$ by Hasse's Theorem) grow similarly, as expected. However, the period for SDLP in finite fields is consistently smaller than the group order, often by a significant factor. This is because the period is bounded by $\phi(p-1)$ multiplied by some constant.

Despite this smaller period size, the SDLP in finite fields consistently requires more computation time than the standard DLP, demonstrating that the additional complexity of group operations in the semidirect product outweighs the advantage of a smaller search space. The DLP in elliptic curves shows the highest computation time despite having a similar order to the finite field case, reflecting the more complex nature of elliptic curve group operations.

For elliptic curves, we have excluded the SDLP from these charts because, as demonstrated in \Cref{sec:elliptic-curve}, the cycle period is bounded by the order of the automorphism, which is at most 24, making the SDLP trivially solvable in constant time.

\subsection{Elementary Abelian Groups}
For elementary abelian groups, we fixed the dimension $n=3$ and varied the prime $p$ from $2^6$ to $2^{12}$. We chose to fix $n=3$ because the problem size grows exponentially with increasing dimension, making experiments with larger dimensions computationally expensive. We specifically examined the impact of the eigenvalue structure of the automorphism matrix on the SDLP complexity. \Cref{fig:elementary-abelian} shows both the cycle period lengths and computation times, comparing matrices with and without eigenvalue 1.

\begin{figure}[h]
\centering
\includegraphics[width=\textwidth]{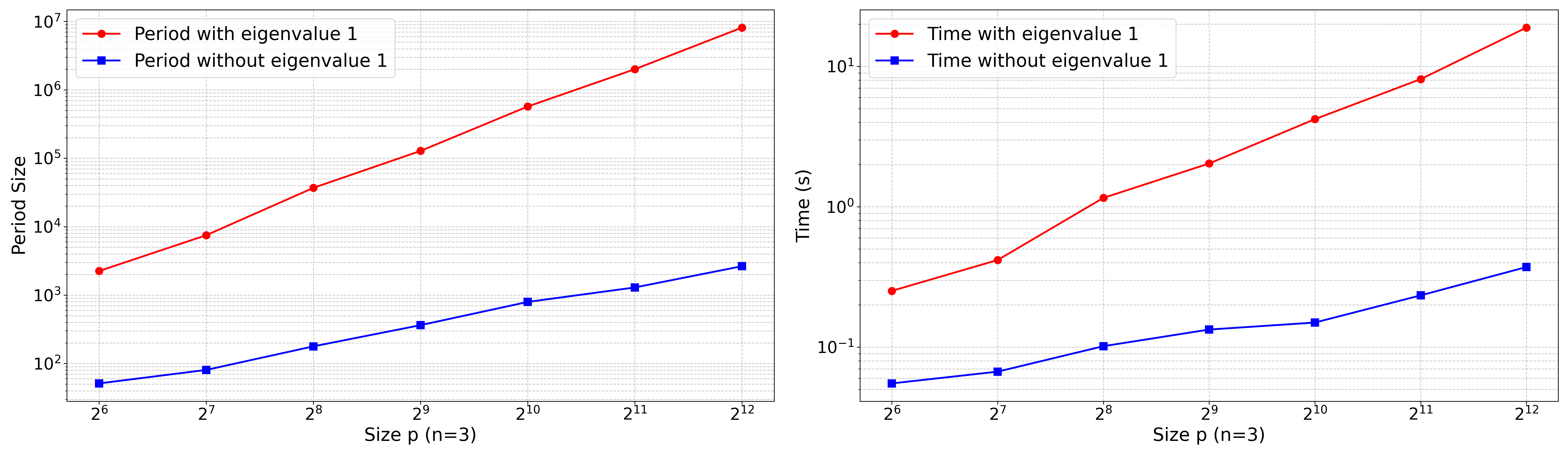}
\caption{Period sizes and computation times for BSGS algorithm on elementary abelian groups ($n=3$) with varying $p$, comparing matrices with and without eigenvalue 1.}
\label{fig:elementary-abelian}
\end{figure}

The experimental results strongly support our theoretical analysis from \Cref{sec:abelian-group}. Our theory predicts that matrices with eigenvalue 1 have larger periods than those without: specifically, the period can reach $p \cdot \ord(A)$ when eigenvalue 1 is present, compared to at most $\ord(A)$ when it is absent. The experiments confirm this distinction - for matrices without eigenvalue 1, the period remains bounded by the order of the matrix, which is at most $p^n-1$. In contrast, matrices with eigenvalue 1 exhibit periods that can be larger by a factor of $p$, reaching close to $p \cdot (p^n-1)$ in many cases. This validates our theoretical complexity bounds of $\mathcal{O}(p^{\frac{n}{2}+1})$ for matrices with eigenvalue 1 versus $\mathcal{O}(p^{\frac{n}{2}})$ for those without. This difference in period length directly translates to computational time, as shown in the right graph of \Cref{fig:elementary-abelian}. The gap between the two cases widens as $p$ increases, with matrices having eigenvalue 1 requiring significantly more computation time.

\section{Conclusion}
In this work, we investigated the classical hardness of the semidirect discrete logarithm problem in finite groups, a problem proposed as a foundation for post-quantum cryptographic protocols. Our investigation reveals that the SDLP does not offer uniform classical hardness advantages over the standard DLP, with the relative difficulty being highly dependent on the underlying platform group.

Our main findings are:
\begin{enumerate}
    \item \textbf{Platform-Dependent Hardness}: The classical complexity of SDLP varies dramatically across different group structures:
    \begin{itemize}
        \item In finite fields $\mathbb{F}_p^\ast$, SDLP has worse worst-case complexity than DLP ($O(p)$ vs $O(\sqrt{p})$). However, in practice with random automorphisms, SDLP instances often have significantly smaller periods due to small values of $\gcd(k-1, p-1)$, making them easier to solve despite the worse theoretical bound.

        \item In elliptic curves $E(\mathbb{F}_p)$, the SDLP degenerates to a trivial problem solvable in constant time, making it strictly easier than the corresponding DLP.
        
        \item In elementary abelian groups $\mathbb{F}_p^n$, the situation reverses: DLP is trivial while SDLP can require up to $O(p^{n/2+1})$ operations when the automorphism matrix has eigenvalue 1, or $O(p^{n/2})$ when it does not.

    \end{itemize}
    
    \item \textbf{Structural Insights}: Our analysis revealed that the non-abelian nature of semidirect products introduces fundamental constraints. We showed that the SDLP can be reformulated as a generalized discrete logarithm problem in the form $u^t v^t = w$, which enabled systematic analysis using adapted classical algorithms.
    
    \item \textbf{Implications for Cryptographic Design}: The assumption that more complex algebraic structures automatically yield harder computational problems is not supported by our findings. The bounded automorphism groups in elliptic curves and the direct solvability in vector spaces demonstrate that additional structure can sometimes facilitate rather than hinder attacks.
\end{enumerate}

Our experimental validation confirms these theoretical insights, with implementations demonstrating the predicted complexity behaviors across all tested platforms. Notably, the presence of eigenvalue 1 in automorphism matrices of elementary abelian groups increases the problem difficulty by exactly the factor predicted by our analysis. Interestingly, while SDLP sometimes has larger asymptotic complexity than DLP, our experiments reveal that the actual cycle periods in SDLP are often significantly smaller than the corresponding group orders in DLP. For instance, in finite fields, even though both problems have comparable worst-case complexity, the SDLP period is typically much smaller than $p-1$ due to small values of $\gcd(k-1, p-1)$ in practice, which can make SDLP instances somewhat easier on average despite similar theoretical bounds.

An important observation about the non-abelian structure is that it creates fundamental obstacles for adapting other classical DLP algorithms. For instance, Pollard's rho algorithm \cite{pollard78dl}, which achieves constant space complexity for standard DLP, cannot be directly adapted to SDLP. In the standard setting, Pollard's rho generates a pseudorandom sequence of group elements of the form $g^{\alpha_i} h^{\beta_i}$ until a collision is detected: $g^{\alpha_i} h^{\beta_i} = g^{\alpha_j} h^{\beta_j}$. Since $h = g^x$ for the unknown discrete logarithm $x$, this collision yields $g^{\alpha_i + x\beta_i} = g^{\alpha_j + x\beta_j}$, which in a cyclic group of order $n$ gives us $\alpha_i + x\beta_i \equiv \alpha_j + x\beta_j \pmod{n}$. Rearranging, we obtain $x(\beta_i - \beta_j) \equiv \alpha_j - \alpha_i \pmod{n}$, allowing us to solve for $x$ when $\gcd(\beta_i - \beta_j, n) = 1$.

However, in the semidirect product setting with our GDLP formulation $u^t v^t = w$, the non-commutativity prevents this approach. A natural attempt would be to generate sequences of the form $u^a v^a w^b$, but upon finding a collision, we cannot rearrange terms to isolate the unknown $t$. The abelian property $g^{\alpha}h^{\beta} = h^{\beta}g^{\alpha}$ that enables the algebraic manipulation in standard Pollard's rho simply does not hold in semidirect products. Moreover, if the semidirect product were abelian—which would enable such rearrangements—the condition $(g_1, \phi_1)(g_2, \phi_2) = (g_2, \phi_2)(g_1, \phi_1)$ would force all automorphisms to act trivially. To see this, consider any $g \in G$ and $\phi \in \Aut(G)$, and let $(g, \text{id})$ commute with $(e, \phi)$. This gives us $(g, \text{id})(e, \phi) = (e, \phi)(g, \text{id})$, which expands to $(g, \phi) = (\phi(g), \phi)$. Therefore $g = \phi(g)$ for all $g \in G$, meaning $\phi$ is the identity automorphism. With only trivial automorphisms, the SDLP reduces to standard DLP since $h = \prod_{i=0}^{t-1} \sigma^i(g) = \prod_{i=0}^{t-1} g = g^t$. This structural constraint suggests that memory-efficient algorithms for SDLP may require fundamentally new approaches beyond adapting existing DLP methods.

These results contribute to the broader understanding of group-based cryptographic hardness assumptions in the post-quantum era. While SDLP may resist quantum attacks in certain settings, our work shows that it does not provide reliable classical hardness advantages across all platforms. This underscores the importance of careful platform selection and thorough classical cryptanalysis when designing cryptographic schemes based on novel group-theoretic problems.

Future research directions include:
\begin{itemize}
    \item Developing novel constant-space algorithms for SDLP that respect the non-abelian structure of semidirect products, potentially achieving better space-time tradeoffs than our BSGS adaptation.
    \item Investigating whether other algebraic structures beyond semidirect products might provide more uniform classical hardness guarantees across different platforms.
    \item Exploring the theoretical limits of classical hardness in non-abelian group settings, particularly understanding when additional structure helps versus hinders cryptanalysis.
\end{itemize}

The balance between structure and hardness revealed in this work highlights the importance of thorough cryptanalysis when proposing new post-quantum primitives based on group-theoretic problems.

\section*{Acknowledgments}
Muhammad Imran is supported by EPSRC through grant number EP/V011324/1.o

\bibliographystyle{splncs04}
\bibliography{myrefs}

\begin{thebibliography}{10}
\providecommand{\url}[1]{\texttt{#1}}
\providecommand{\urlprefix}{URL }
\providecommand{\doi}[1]{https://doi.org/#1}

\bibitem{battarbee2024semidirect}
Battarbee, C., Borin, G., Brough, J., Cartor, R., Hemmert, T., Heninger, N.,
  Jao, D., Kahrobaei, D., Maddison, L., Persichetti, E., et~al.: On the
  semidirect discrete logarithm problem in finite groups. In: International
  Conference on the Theory and Application of Cryptology and Information
  Security. pp. 330--357. Springer (2024)

\bibitem{battarbee2023spdh}
Battarbee, C., Kahrobaei, D., Perret, L., Shahandashti, S.F.: Spdh-sign:
  towards efficient, post-quantum group-based signatures. Cryptology ePrint
  Archive  (2023)

\bibitem{buhler93nfs}
Buhler, J.P., Lenstra, H.W., Pomerance, C.: Factoring integers with the number
  field sieve. In: Lenstra, A.K., Lenstra, H.W. (eds.) The development of the
  number field sieve. pp. 50--94. Springer Berlin Heidelberg, Berlin,
  Heidelberg (1993)

\bibitem{castryck2018csidh}
Castryck, W., Lange, T., Martindale, C., Panny, L., Renes, J.: Csidh: an
  efficient post-quantum commutative group action. In: Advances in
  Cryptology--ASIACRYPT 2018: 24th International Conference on the Theory and
  Application of Cryptology and Information Security, Brisbane, QLD, Australia,
  December 2--6, 2018, Proceedings, Part III 24. pp. 395--427. Springer (2018)

\bibitem{childs2014quantum}
Childs, A.M., Ivanyos, G.: Quantum computation of discrete logarithms in
  semigroups. Journal of Mathematical Cryptology  \textbf{8}(4),  405--416
  (2014)

\bibitem{couveignes2006hard}
Couveignes, J.M.: Hard homogeneous spaces. Cryptology ePrint Archive  (2006)

\bibitem{darafsheh2005}
Darafsheh, M.: Order of elements in the groups related to the general linear
  group. Finite Fields and Their Applications  \textbf{11}(4),  738--747
  (2005). \doi{https://doi.org/10.1016/j.ffa.2004.12.003},
  \url{https://www.sciencedirect.com/science/article/pii/S1071579704000760}

\bibitem{gordon93nfs}
Gordon, D.M.: Discrete logarithms in \$gf ( p )\$ using the number field sieve.
  SIAM Journal on Discrete Mathematics  \textbf{6}(1),  124--138 (1993).
  \doi{10.1137/0406010}, \url{https://doi.org/10.1137/0406010}

\bibitem{grigoriev2014tropical}
Grigoriev, D., Shpilrain, V.: Tropical cryptography. Communications in Algebra
  \textbf{42}(6),  2624--2632 (2014)

\bibitem{habeeb2013public}
Habeeb, M., Kahrobaei, D., Koupparis, C., Shpilrain, V.: Public key exchange
  using semidirect product of (semi) groups. In: Applied Cryptography and
  Network Security: 11th International Conference, ACNS 2013, Banff, AB,
  Canada, June 25-28, 2013. Proceedings 11. pp. 475--486. Springer (2013)

\bibitem{imran2024efficient}
Imran, M., Ivanyos, G.: Efficient quantum algorithms for some instances of the
  semidirect discrete logarithm problem. Designs, Codes and Cryptography pp.
  1--19 (2024)

\bibitem{lenstra90nfs}
Lenstra, A.K., Lenstra, H.W., Manasse, M.S., Pollard, J.M.: The number field
  sieve. In: Proceedings of the Twenty-Second Annual ACM Symposium on Theory of
  Computing. p. 564–572. STOC '90, Association for Computing Machinery, New
  York, NY, USA (1990). \doi{10.1145/100216.100295},
  \url{https://doi.org/10.1145/100216.100295}

\bibitem{pollard78dl}
Pollard, J.M.: Monte carlo methods for index computation $(\operatorname{mod}
  p)$. Mathematics of Computation  \textbf{32}(143),  918--924 (1978),
  \url{http://www.jstor.org/stable/2006496}

\bibitem{Shanks1971ClassNA}
Shanks, D.: Class number, a theory of factorization, and genera. In:
  Proceedings of Symposia in Pure Mathematics. pp. 415--440 (1971),
  \url{10.1090/pspum/020/0316385}

\bibitem{shor1994algorithms}
Shor, P.W.: Algorithms for quantum computation: discrete logarithms and
  factoring. In: Proceedings 35th annual symposium on foundations of computer
  science. pp. 124--134. Ieee (1994)

\bibitem{silverman2009arithmetic}
Silverman, J.H.: The arithmetic of elliptic curves, vol.~106. Springer (2009)

\bibitem{stickel2005new}
Stickel, E.: A new method for exchanging secret keys. In: Third International
  Conference on Information Technology and Applications (ICITA'05). vol.~2, pp.
  426--430. IEEE (2005)

\end{thebibliography}

\end{document}